\documentclass[journal]{IEEEtran}
\usepackage{algpseudocode}
\usepackage{algorithm}
\usepackage{algorithmicx}
\usepackage{lipsum} 
\usepackage{xcolor}
\newenvironment{MyColorPar}[1]{%
	\leavevmode\color{#1}\ignorespaces%
}{%
}%
\usepackage{todonotes}
\usepackage{epsf}
\usepackage{epsfig}
\usepackage{times}
\usepackage{epsfig}
\usepackage{graphicx}
\usepackage{bbold}
\usepackage{mathtools}
\usepackage{mathrsfs}
\usepackage{amssymb}
\usepackage{pdfpages}
\usepackage{epstopdf}
\usepackage{float}
\newfloat{algorithm}{t}{lop}
\usepackage{amsmath}
\usepackage{dsfont}
\usepackage{relsize}
\usepackage{lettrine} 
\usepackage{amsmath,epsfig,amssymb,algorithm,algpseudocode,amsthm,cite,url}
\usepackage{xfrac}

\usepackage{subcaption}
\allowdisplaybreaks
\usepackage{csquotes}

\usepackage{verbatim}
\usepackage[english]{babel}
\usepackage{amsmath,amssymb}

\usepackage[justification=centering]{caption}
\usepackage{verbatim}

\newtheorem{theorem}{\bf Theorem}

\newtheorem{remark}{Remark}

\begin{document}
	\title{\huge Sum-Rate Analysis for High Altitude Platform (HAP) Drones with Tethered Balloon Relay}\vspace{0.1cm}


 \author{Sudheesh P. G,
    	Mohammad Mozaffari,~\IEEEmembership{Student Member,~IEEE,}
    	Maurizio Magarini,~\IEEEmembership{Member,~IEEE,}
    	Walid Saad,~\IEEEmembership{Senior Member,~IEEE,}
    	and P. Muthuchidambaranathan
\thanks{Sudheesh P. G and P. Muthuchidambaranathan are with Department of Electronics and Communication Engineering, National Institute of Technology, Tiruchirappalli, India (e-mail: pgsudheesh@gmail.com; muthuc@nitt.edu).}
\thanks{Mohammad Mozaffari and Walid Saad are with Wireless@VT, Electrical and Computer
	Engineering Department, Virginia Tech, Blacksburg, VA, 24061, USA, (e-mail:
	mmozaff@vt.edu; walids@vt.edu).}
\thanks{Maurizio Magarini is with Dipartimento di Elettronica, Informazione e Bioingegneria Politecnico di Milano, 20133 Milano, Italy (e-mail: maurizio.magarini@polimi.it).} \vspace{-0.5cm}}

	\maketitle
\begin{abstract}
	\textcolor{black}{High altitude platform (HAP) drones can provide
	broadband wireless connectivity to ground users in rural areas  by establishing line-of-sight (LoS) links and exploiting
	effective beamforming techniques.} However, at high altitudes, acquiring the channel state information (CSI) for HAPs, which is a key component to perform beamforming, \textcolor{black}{is challenging.} In this paper, by exploiting an interference alignment (IA) technique, a novel method for achieving the maximum sum-rate in HAP-based communications without CSI is proposed. In particular, to realize IA, a multiple-antenna tethered balloon is used as a relay between multiple HAP drones and ground stations (GSs). Here, a multiple-input multiple-output X network  system is considered. The capacity of the considered $M \times N$ X network with a tethered balloon relay is derived in closed-form. Simulation results corroborate the theoretical findings and show that the proposed approach yields the maximum sum-rate in multiple HAPs-GSs communications in absence of CSI. The results also show the existence of an optimal balloon's altitude for which the sum-rate is maximized.\vspace{-0.2cm}
\end{abstract}

\section{Introduction}\vspace{-0.01cm}
Satellite and terrestrial cellular communications are the two most widely used
communication systems for providing global connectivity to mobile ground users.
While satellites can deliver wireless service to users in remote areas, their spectral efficiency is limited by their large footprints \cite{mohammad}. Meanwhile, terrestrial communication systems cannot
guarantee a reliable service for users in remote, rural areas, due to the lack of infrastructure nodes, such as base stations (BSs). High altitude platform (HAP) drones can substantially extend the coverage of terrestrial networks
by establishing line-of-sight (LoS) links and adjusting their altitude \cite{mozaffaricoverage}, \cite{Letter}.

To exploit the spatial dimension and enhance spectral efficiency,
	HAPs will typically rely on highly directive antennas to communicate with ground stations \cite{liugrace2}.
     \textcolor{black}{In a single HAP system with multiple antennas at the transmitter, a spatial multiplexing gain cannot be typically achieved due to a high correlation between parallel paths \cite{madhow2}.
     However, the deployment of multiple spatially separated HAPs can be a promising solution to exploit spatial multiplexing and  boost spectral efficiency. In particular, by using a large number of antennas at the HAPs,  one can provide a precise beamforming  which is a key requirement for spatial multiplexing.
   To this end, channel state information (CSI) at the transmitter (CSIT)  
  is required \cite{madhow2}, \cite{jafarXlimit}.
   However, in HAP drone systems, acquiring precise CSIT is challenging due to the high altitudes and the movement of the drones. Consequently, in HAPs-to-ground stations (GSs) communications\footnote{An HAP drones-GSs wireless system in which each HAP
   	drone carries a dedicated symbol for each GS can be modeled
   	as an X network. The X network houses all possible channel
   	models such as the interference channel, the Y channel,
   	and the Z channel. Beyond offering a generalized structure, an
   	X network offers a maximum capacity as compared to other
   	channel models such as the interference channel \cite{jafarXlimit}.}  exploiting spatial multiplexing, which can yield a maximum possible sum-rate, is also challenging.}\\    
\indent\textcolor{black}{  One practical approach to achieve a maximum possible sum-rate for HAP drones-to-GSs communications is via the use of interference alignment (IA) schemes  \cite{jafarXlimit,Pantisano,cadambe}.
	In particular, at a high {SNR} regime, which is  a typical case in HAP communications, IA can achieve a maximum sum-rate by restricting
	interference beams to a smaller subspace that does not overlap
	with the desired signal space \cite{jafarXlimit}.}
\textcolor{black}{Unlike terrestrial wireless systems in which the BSs' positions are fixed,
    acquiring CSIT to implement IA in HAPs-GSs communication systems is challenging due to imperfect HAP drone stabilization.
	As a result, the lack of exact CSI at HAP drones
    can yield a significant degradation of the sum-rate performance.}
    Nevertheless, with the use of relays, it is possible to achieve maximum sum-rate
 when CSIT is not available \cite{yener2}.\\
 \indent Unlike the time domain realization of IA,  the performance of IA in the spatial domain is limited due to the restrictions in designing precoding matrices\cite{jafarXlimit}. Therefore, the maximum possible sum-rate can be achieved by implementing relay-assisted IA in the time domain.
	In this case, one can use popular relaying mechanisms 
such as amplify and forward (AF) or decode and forward (DF) \cite{yener2}. 
	 An AF relaying scheme with multiple
    relays achieves an upper bound for the $\mathrm{DoF}$
	of a generalized X network \cite{yener2}.
    Hence, by adopting an AF relaying scheme such as the one in \cite{yener2}, multiple relays can be used to
    achieve maximum possible sum-rate of the system where HAP drones have knowledge of CSI.
    Similar results can be achieved with a DF-based relaying scheme with single relay in
	a two-user X channel \cite{iran}. However, the previous works in \cite{yener2} and \cite{iran} did not investigate the use of a DF relaying mechanism for IA in an HAP drones system.\\ 
\indent\textcolor{black}{The main contribution of this paper is a novel framework for maximizing the sum-rate of a relay-aided HAP drones wireless system when the CSI is not available.
	In particular, to achieve the maximum sum-rate, we propose a
 DF scheme involving $M$ HAP drones, $N$ ground receivers,
  and one relay with $(M-1)\times (N-1)$ antennas.
    In this scenario, we show that it is possible to achieve the maximum possible sum-rate by exploiting the IA scheme.
    Moreover, we derive a closed-form analytical expression for
     the capacity of an $M \times N$ X channel with tethered balloon relay. 
    Simulation results verify our analytical results and show that a significant sum-rate gain can be achieved by using the proposed scheme.}

	\section{System Model}
	Consider a geographical area with $N$
	GSs (or receivers) and a tethered balloon attached to a control station, as shown in Fig.\,\ref{fig:twousereps}.
	This control station provides the power required to operate the tethered balloon.
	Meanwhile, the GSs receive data from $M$ HAPs,
	which are located at altitudes within the range of 17-22\,km  \cite{mohammad}.
	Each HAP and receiver houses $A$ antennas while the tethered balloon
	has $(M-1)\times (N-1)$ antennas. Each GS receives data from
	each HAP, forming an $M \times N$ X network. Unlike \cite{liugrace2}, which uses a frequency
    duplexing technique to avoid the interference in HAP communications,
	our proposed model uses a relay that operates in half duplex mode in the same frequency band.	\vspace{-0.1cm}

\subsection{Channel Model}
	For terrestrial communications, the channel is typically modeled as
	Rayleigh in urban areas and Rician in suburban scenarios.
	However, in air-to-ground communications, the channel has different characteristics \cite{zajic}, \cite{LetterOT}.
	In urban environments, the air-to-ground channel experiences Rician
	fading due to the presence of LoS links.
    In suburban areas, a Rayleigh fading is experienced
	due to the presence of reflected signals which are stronger than LoS signals \cite{zajic}.

	Here, we adopt a Rician channel model
	in which both LoS and non-LoS (NLoS) paths are considered.
	Therefore, the channel gain matrix can be represented as \cite{HAPcapacity}:
	\begin{equation}\label{Hmat}
		\boldsymbol{{H}}= \sqrt{\frac{\kappa}{1+\kappa}}\boldsymbol{{H}}_\textrm{LoS} + \sqrt{\frac{1}{1+\kappa}}\boldsymbol{{H}}_\textrm{NLoS},
	\end{equation}
	where $\boldsymbol{{H}}_\textrm{LoS}$ and $\boldsymbol{{H}}_\textrm{NLoS}$
 represent, respectively, the channel matrices for LoS and non-LoS communication.
	The Rician factor $\kappa$
	is given by,
		$\kappa=\frac{\sigma^{2}_\textrm{LoS}}{\sigma^{2}_\textrm{NLoS}}$ \cite{matlabbook},
		where $\sigma^{2}_\textrm{LoS}$ and $\sigma^{2}_\textrm{NLoS}$ are the power
	of LoS path and NLoS path, respectively.
	For our model, we consider $A$ antennas at
    each HAP and GS. The role of HAP and GS as transmitter or receiver
    can be reversed, using the reciprocity property. The static MIMO channel, excluding the path loss, is given by \cite{matlabbook}:\vspace{-0.2cm}

\begin{equation}
	\hspace{-.15cm}	\boldsymbol{{\bar{H}}}_\textrm{LoS} \hspace{-.1cm}= \hspace{-.1cm} \hspace{-.1cm} \begin{bmatrix}
			1 \\
			e^{j2\pi\frac{d_{R}}{\lambda}}\hspace{-.05cm}\sin({\theta_A}) \\
			\vdots \\
			e^{j2\pi \frac{d_R}{\lambda}(M-1)}\hspace{-.05cm}\sin(\theta_A)
		\end{bmatrix}\hspace{-.03cm}\cdot\hspace{-.03cm}
		\begin{bmatrix}
			1 \\
			e^{j2\pi\frac{d_{T}}{\lambda}}\hspace{-.05cm}\sin(\theta_D) \\
			\vdots \\
			e^{{j2\pi \frac{d_T}{\lambda}}(N-1)}\sin(\theta_D)
		\end{bmatrix}^{T}\hspace{-.3cm},
	\end{equation}
\noindent \textcolor{black}{where $d_{R}$ and $d_{T}$ are the
	antenna spacing at the  receiver and transmitter ($L \gg d_{R}, d_{T}$)}, and $\lambda$ is the wavelength.
    Also, $\theta_A$ and $\theta_D$
	represent, respectively, the angle-of-arrival at the receiver and angle-of-departure at the transmitters. We also note that the NLoS MIMO channel	follows a Rayleigh distribution.\vspace{-0.00cm}
		\begin{figure}[!t]
			\begin{center}
				\vspace{-0.1cm}
				\includegraphics[width=6.9cm]{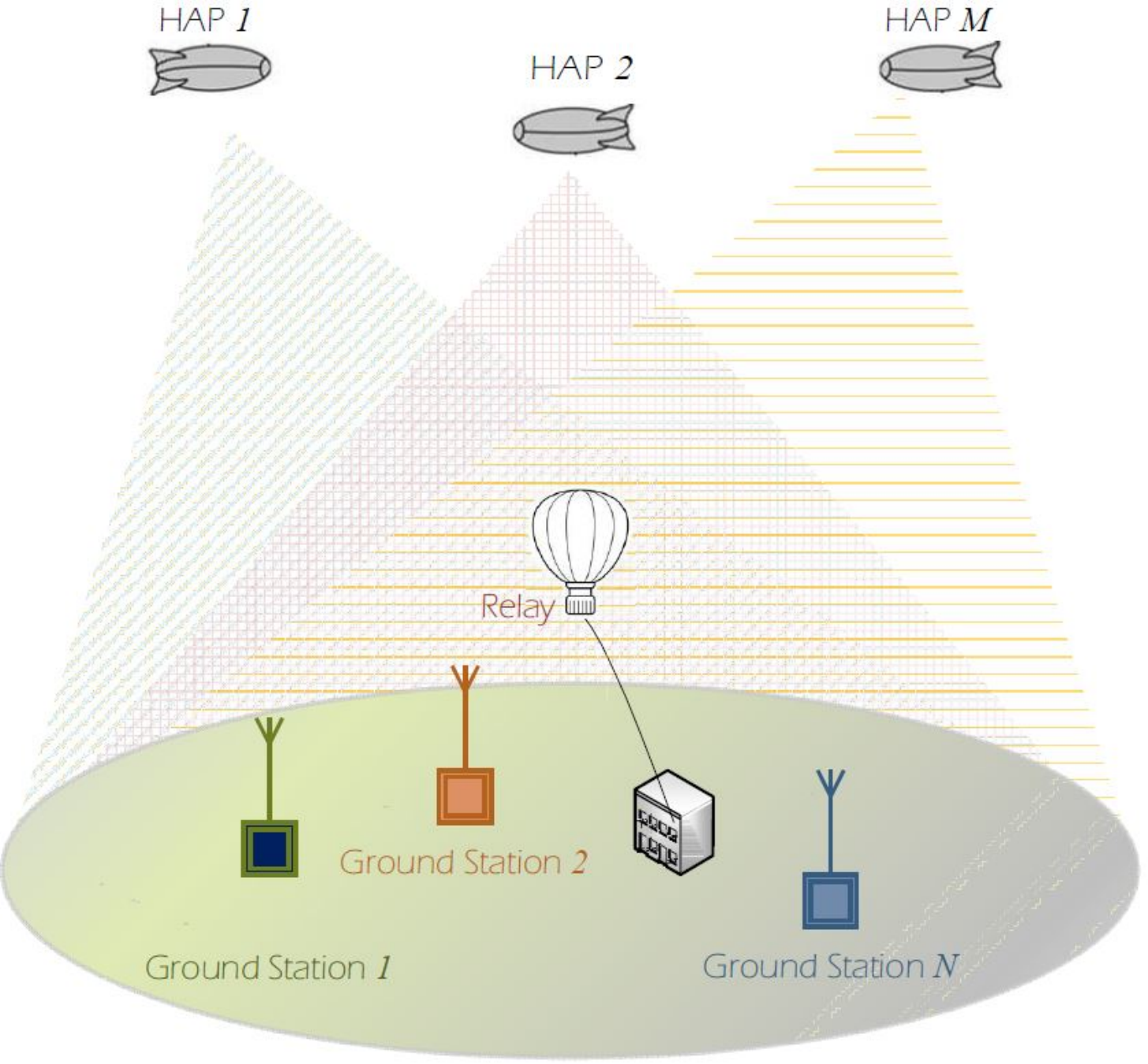}
				\vspace{-0.01cm}
				\caption{System model. \vspace{-.45cm}}
				\label{fig:twousereps}
			\end{center}
		\end{figure}
	\section{DF relay aided Interference Alignment for
		systems without CSIT}
    Knowledge of CSIT is a key requirement for designing a precoder at the HAPs.
    However, at high altitudes acquiring a precise CSIT is challenging due to the
    difficulty in stabilizing the aerial platform that is affected by wind and other natural factors.
    In such a case, to achieve IA, a feasible solution is to
    accommodate a tethered balloon relay to perform DF operation. A DF scheme will require at least $M$ antennas
    to decode $M$ transmitted symbols.
    However, the DF relay must accommodate $(M-1) \times (N-1)$ antennas to perform IA.
    By exploiting temporal domain characteristics, we can transmit
	$MN$ number of symbols  in $M+N-1$ time slots for
    transmitter and receiver. \vspace{-0.20cm}
    \subsection{Feasibility of DF in tethered balloon relay}
    For IA, the tethered balloon must have $(M-1) \times (N-1)$ antennas.
    To tightly pack these antennas in a tethered balloon relay,
    the HAPs must be placed sufficiently apart 
    so that the links from HAPs
    to GSs become uncorrelated.
     We deduce the minimal separation required between HAPs,
     when they are located at 18 km above earth, operating at 48 GHz \cite{mohammad}.
     From the concept of LoS MIMO \cite{madhow2}, we have:
     \vspace{0.02cm}
    \begin{equation}\label{madhow}
     d_\textrm{HAP}d_\textrm{GS}=\frac{L \lambda}{\beta},
    \end{equation}
   \noindent where $L$ is the distance between an HAP and the GS,
    $\beta$ represents the degrees of freedom.
    Also, \textcolor{black}{$d_\textrm{HAP}$ and $d_\textrm{GS}$, respectively, represent the inter-HAP and inter-GS distances}.
    To determine the HAP separation distance, we choose the following parameters:
    $d_\textrm{GS}= 0.1$\,m, $L= 18$\,km, $\lambda= 0.00625$\,m, and $\beta=1$. Now, by using (\ref{madhow}),
    we find that the HAPs must be spaced 1125\,m apart to ensure that the channels will be uncorrelated. In this case, each HAP drone-GS wireless communication channel will be full rank which allows sending data over multiple paths.

	The single tethered balloon relay performs the DF operation
	and operates in half-duplex mode as shown in \cite{yener2}.
    The communication between HAPs and GSs is carried out in two phases,
	direct transmission and relay aided transmission.
	In the first phase, GSs receive signal directly
	from the HAP drones while the relay remain silent.
	In this case, HAP drones transmit data to
	one specific GS during one time slot.
   In the second phase, tethered balloon relay is active and transmits data to the GS after precoding.
	Therefore, the GS receives signals from the relay and the first transmitter. \vspace{-0.2cm}

\subsection{Capacity of Rician X network}
The asymptotic sum-rate of a network as a function of the signal-to-noise-ratio ({SNR})
can be expressed as \cite{jafarXlimit,cadambe}:
\begin{equation}\label{capeq}
C=\beta\cdot \log(\gamma)+\mathcal{O}(\log(\gamma)),
\end{equation}
where $\gamma$ is the SNR value at a given receiver.

The $\mathrm{DoF}$ for an X network with $M$ transmitters and $N$ receivers
each with $A$ antennas, is equal to $\beta=\frac{MNA}{M+N-1}$.
However, we consider a relay-aided system that uses a DF relaying mechanism.
\begin{MyColorPar}{black}
We denote the channel matrices between the relay and HAP $i$ by $\boldsymbol{H}_{i}$, and between GS $j$ and the relay by $\boldsymbol{G}_{j}$.
Since the Rician factor of the HAP-to-relay link is greater than
that of the relay-to-GS link, we model $\boldsymbol{H}_{i}$ and
$\boldsymbol{G}_{j}$ with different $\kappa$ values. In this case, the small-scale fading matrices $\boldsymbol{\bar{H}}_{i}$ and $\boldsymbol{\bar{G}}_{j}$ are
obtained by replacing $\kappa$ in (1) by $\kappa_{i}^{u}$ and $\kappa_{j}^{l}$,
respectively.
After adding the path loss components to $\boldsymbol{\bar{H}}_{i}$ and $\boldsymbol{\bar{G}}_{j}$, we get $\boldsymbol{H}_{i}=\frac{\alpha_{i}}{(d_{Ri})^{2}}\boldsymbol{\bar{H}}_{i}$
and $\boldsymbol{G}_{j}=\frac{\psi_{j}}{(d_{jR})^{2}}\boldsymbol{\bar{G}}_{j}$,
where $\alpha_{i}$ and $\psi_{j}$
represent, respectively, the channel gains in $\boldsymbol{H}_{i}$ and $\boldsymbol{G}_{j}$
 at a 1\,m reference distance. Also, $d_{Ri}$ and $d_{jR}$ are the
link distances in HAP $i$-tethered balloon and tethered balloon-GS $j$ communications.
In order to find  capacity of the considered relay aided HAPs-GSs system,
we first calculate the {SNR} of each stream.
Here, we have two phases: HAPs-tethered balloon communications, and tethered balloon-GSs communications.
We begin with the first phase, and the same steps can be used to
find the sum-rate of the second phase.
The zero-forcing detection (ZF) SNR of the $k^\textrm{th}$ parallel channel, $\mathrm{\gamma}_{k}$, is given by \cite{HAPcapacity}:\vspace{-0.05cm}
\begin{equation}\label{snr1}
\mathrm{\gamma}_{k}=\frac{\Gamma}{[\boldsymbol{{W}}^{-1}]_{k,k}},
\end{equation}
where $\Gamma$ is the transmit power per symbol, and
$\boldsymbol{{W}}=\boldsymbol{{H}}^{H}\boldsymbol{{H}}$.
Now, given the knowledge of ZF SNR, we derive the explicit expression
for ZF-capacity as follows.
\vspace{-0.02cm}
\begin{theorem}
	\label{theorem1}\normalfont
The capacity of a relay aided multiple HAPs-GSs communication is given by:
\begin{align}\label{capprop}
C=\frac{MN}{M+N-1}\textrm{min}\Big(&\sum\limits_{i=1}^{M}\log_{2}
\big(1+\Gamma_{s}\boldsymbol{h}_{i1}^{H}\boldsymbol{Q}_{i}\boldsymbol{h}_{i1}\big), \nonumber \\
&\sum\limits_{j=1}^{N}\log_{2}
\big(1+\Gamma_{s}\boldsymbol{g}_{j1}^{H}\boldsymbol{W}_{j}\boldsymbol{g}_{j1}\big) \Big),
\end{align}
\vspace{-0.1cm}
where $\boldsymbol{Q}_{i}$ and $\boldsymbol{W}_{j}$ are given by: \vspace{0.08cm}
\begin{align}
   \boldsymbol{Q}_{i}&=  [\boldsymbol{I}-\boldsymbol{\tilde{H}}_{i}(\boldsymbol{\tilde{H}}_{i}^{H}
   \boldsymbol{\tilde{H}}_{i})^{-1}\boldsymbol{\tilde{H}}_{i}^{H}],\\
   \boldsymbol{W}_{j}&=  [\boldsymbol{I}-\boldsymbol{\tilde{G}}_{j}(\boldsymbol{\tilde{G}}_{j}^{H} \boldsymbol{\tilde{G}}_{j} )^{-1}\boldsymbol{\tilde{G}}_{j}^{H}],
\end{align}
with $\boldsymbol{\tilde{H}}_{i}$, $\boldsymbol{h}_{i1}$ and $\boldsymbol{\tilde{G}}_{j}$, $\boldsymbol{g}_{j1}$
 obtained from channel matrices $\boldsymbol{H}_{i}$  and $\boldsymbol{G}_{j}$ based on $\boldsymbol{H}_{i}=[\boldsymbol{h}_{i1} \boldsymbol{\tilde{H}}_{i}]$ and
$\boldsymbol{G}_{j}=[\boldsymbol{g}_{j1} \boldsymbol{\tilde{G}}_{j}]$.
\end{theorem}
\begin{proof}
In the proposed model, we have two phases: HAPs-tethered balloon and tethered balloon-GSs communications. We first
find the SNR of HAP $i$-to-relay link, $\mathrm{\gamma}_{i}$, and
relay-to-GS $j$, $\mathrm{\gamma}_{j}$. Then, we proceed to calculate the sum-rate.
In general, the SNR of the first stream at the receiver is given by: \vspace{-0.2cm}
\begin{equation}\label{SNR1}
\mathrm{\gamma}_{1}=\frac{   \frac{E_{s}}{\sigma^{2}N_{T}}} {\big[(\boldsymbol{{H}}^{H}\boldsymbol{{H}})^{-1}\big]_{1,1}},
\end{equation}
where $\frac{E_{s}}{N_{T}}$ is the transmitted energy per symbol, and $\sigma^{2}$ is the noise power.
From \cite{matrixinv}, we know that $[\boldsymbol{{W}}^{-1}]_{1,1}$ can be found from the elements of
$\boldsymbol{{H}}$. That is,
\begin{equation}\label{winv11}
{\big[(\boldsymbol{{H}}^{H}\boldsymbol{{H}})^{-1}\big]_{1,1}}=
\boldsymbol{h}_{1}^{H}[\boldsymbol{I}-\boldsymbol{\tilde{H}}(\boldsymbol{\tilde{H}}^{H}
   \boldsymbol{\tilde{H}})^{-1}\boldsymbol{\tilde{H}}^{H}]\boldsymbol{h}_{1},
\end{equation}
where, $\boldsymbol{H}=[\boldsymbol{h}_{1} \boldsymbol{\tilde{H}}]$.
Now, substituting (\ref{winv11}) in (\ref{SNR1}) leads to:
\begin{equation}\label{SNRtheory}
\mathrm{\gamma}_{1}=  \frac{E_{s}}{\sigma^{2}N_{T}}
\boldsymbol{h}_{1}^{H}[\boldsymbol{I}-\boldsymbol{\tilde{H}}(\boldsymbol{\tilde{H}}^{H}
   \boldsymbol{\tilde{H}})^{-1}\boldsymbol{\tilde{H}}^{H}]\boldsymbol{h}_{1}.
\end{equation}
 By substituting $\boldsymbol{H}_{i}$ for  $\boldsymbol{H}$
in (\ref{SNRtheory}), we get $\mathrm{\phi}_{i}$:
\begin{equation}\label{SNRtheory2}
\mathrm{\phi}_{i}=  \frac{{E}_\textrm{HAP}}{\sigma^{2}N_{T}}
\boldsymbol{h}_{i1}^{H}[\boldsymbol{I}-\boldsymbol{\tilde{H}}_{i}(\boldsymbol{\tilde{H}}_{i}^{H}
   \boldsymbol{\tilde{H}}_{i})^{-1}\boldsymbol{\tilde{H}}_{i}^{H}]\boldsymbol{h}_{i1},
   \end{equation}
   where $\mathds{E}_\textrm{HAP}$ is the transmitted power of each HAP.
Similarly, the SNR in GS $j$ is:
\begin{equation}\label{SNRtheory3}
\mathrm{\varphi}_{j}= \frac{E_{\textrm{BL}}}{\sigma^{2}N_{T}}
\boldsymbol{g}_{j1}^{H}[\boldsymbol{I}-\boldsymbol{\tilde{G}}_{j}(\boldsymbol{\tilde{G}}_{j}^{H}
   \boldsymbol{\tilde{G}}_{j})^{-1}\boldsymbol{\tilde{G}}_{j}^{H}]\boldsymbol{g}_{j1},
   \end{equation}
   where $E_{\textrm{BL}}$ is the transmitted power of tethered balloon, associated with $\boldsymbol{G}_{j}$.
The overall system capacity can be defined as in \cite{afdf}:
\begin{equation}\label{DFcapacity}
C_{DF}(\gamma)=\beta \cdot \textrm{min}({C_{1}(\gamma),C_{2}(\gamma)}),
\end{equation}
where $C_{1}=\sum_{i=1}^{M‎}\log(1+\phi_{i})$ is the sum-rate between the HAPs-tethered balloon link and
$C_{2}=\sum_{j=1}^{N}\log(1+\varphi_{j})$ is the sum-rate between the tethered balloon-GSs link.
Finally, by substituting    (\ref{SNRtheory2})
 and (\ref{SNRtheory3}) in (\ref{DFcapacity}) this theorem is proved.
\vspace{-0.1cm}
\end{proof}
\vspace{-0.01cm}
Using Theorem\,\ref{theorem1}, it is possible to analyze
the impact of $\kappa_{i}^{u}$, $\kappa_{j}^{l}$, $d_{Ri}$, and $M$ on the capacity of the Rician X-channel in the HAP drones wireless system.
\begin{remark}
	At high $\textnormal{SNR}$, the system fails to achieve a higher sum-rate.
	This is because, at larger $\kappa_{i}^{u}$ and $\kappa_{j}^{l}$, the columns of $\boldsymbol{{H}}$ will be correlated.
	Hence, IA will fail to achieve a maximum capacity for higher $\kappa_{i}^{u}$ and $\kappa_{j}^{l}$ values. \vspace{-0.1cm}
\end{remark}
\begin{remark}
	When $M=N$ and $\kappa_{i}^{u}$=$\kappa_{j}^{l}$, the capacity of the system is maximum
	when $\frac{E_{\textrm{HAP}}\times(d_{Ri})^{2}}{\alpha_{i}}=\frac{E_{\textrm{BL}}\times(d_{jR})^{2}}{\psi_{j}}$.
	Also, when the transmitted power $E_{\textrm{HAP}}=E_{\textrm{BL}}$, the maximum capacity is obtained
	when the tethered balloon is at the center of the HAPs-GSs link.
\end{remark}

 \vspace{-0.1cm}

\section{Numerical Results}
   Here, we evaluate the sum-rate performance of the relay aided multiple HAPs-to-GSs communications. HAPs are located at an altitude of 18\,km above earth
   and each GS receives signals from all HAPs. 
    We consider an
    X channel, consisting of $M= 2$ and $3$ HAPs, a tethered balloon, and 3 ground stations.
\begin{figure}[!t]
    \centering

    \includegraphics[width=0.98\linewidth]{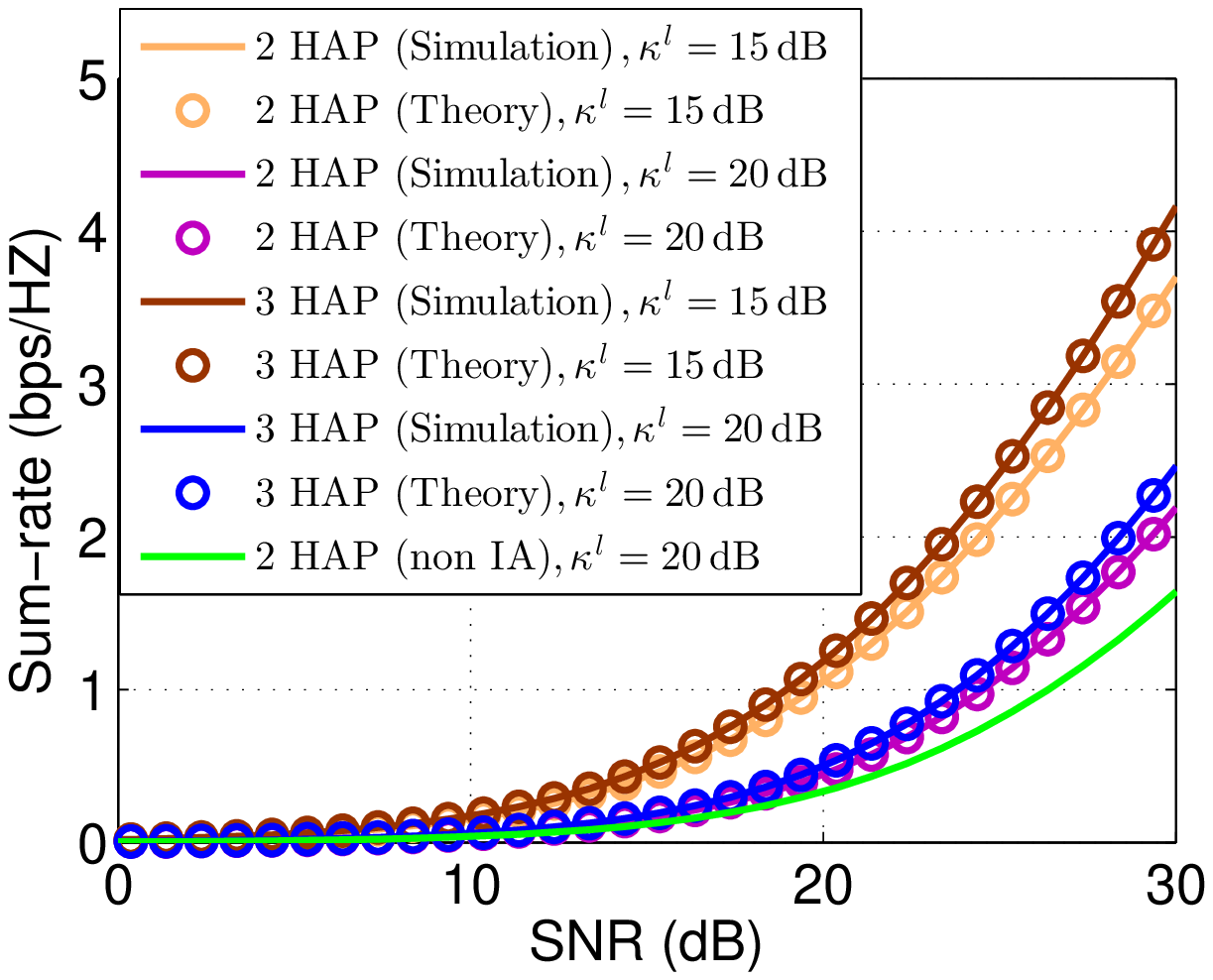}
    \caption{Sum-rate vs. SNR.}
   \label{fig:HAPeffect}
\vspace{0.4cm}
     \includegraphics[width=0.96\linewidth]{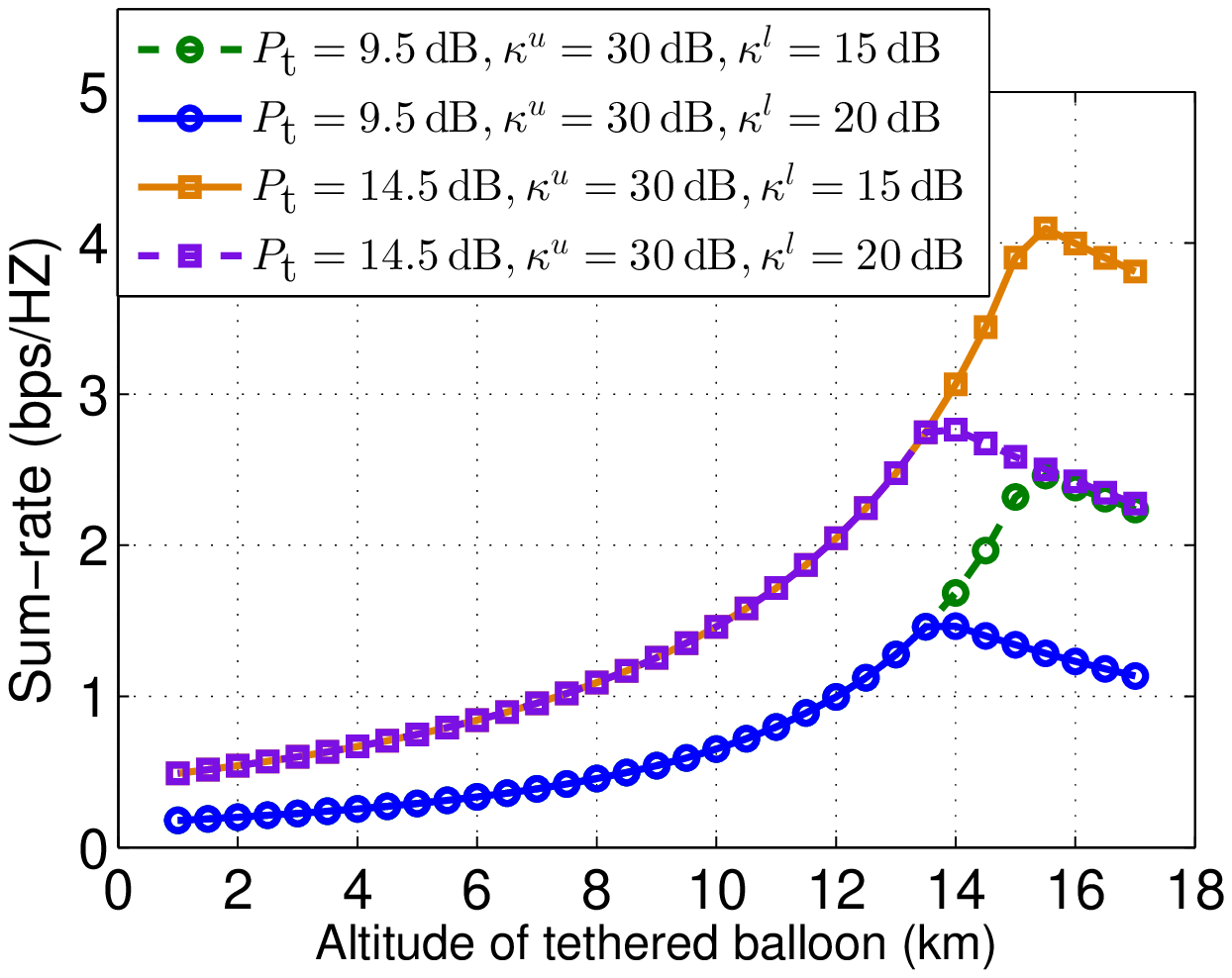}
    \caption{Sum-rate vs. tethered balloon's altitude. }\vspace{-0.2cm}
   \label{fig:capacity}
    \end{figure}
Fig.\,\ref{fig:HAPeffect} shows the sum-rate of the Rician channel as a function SNR for different Rician factors and number of HAPs.
    As we can see from this figure, the sum-rate obtained using simulations (the curves shown without a circle marker) is in agreement with the analytical derivation.
    Without the use of a tethered balloon, the sum-rate of the considered system decreases as the maximum DoF cannot be achieved. 
   While exploiting IA, however, the system is equivalent
    to a system in which CSI is known at the transmitters and, hence, a higher sum-rate is achieved.
    In this analysis, we consider $\kappa^{u}=30$\,dB and $d_{SR}=1$ km
    and we vary $\kappa^{l}$ and $M$.
    As expected, the sum-rate increases by using the IA solution
    and decreasing the Rician factor, $\kappa^{l}$.
    For instance, at $\textrm{SNR}=25$\,dB, the channel with $\kappa^{l}= 20$\,dB
    offers a sum-rate gain of 34.8\%. Interestingly, under the same {SNR}, the sum-rate increases by up to 1.7 times when
    $\kappa^{l}=15$\,dB.
    From Fig.\,\ref{fig:HAPeffect}, we can also see that, as $\kappa$ increases,
    the sum-rate degradation occurs at a higher {SNR}.
    This is due to the fact that increasing $\kappa$ increases the correlation between channels,
    which results in a degradation of the asymptotic performance.

    From Fig.\,\ref{fig:capacity}, we can see that sum-rate
    increases as the tethered balloon moves away from the HAPs.
    Similarly, the sum-rate decreases with $\kappa^{l}$, provided that the HAPs are located sufficiently apart.
     This is due to the fact that exploiting IA while deploying a tethered balloon relay allows achieving a
     maximum DoF in the system.
     Fig.\,\ref{fig:capacity} shows that the optimal altitude of
     tethered balloon at which the HAPs-tethered balloon-GSs link
     has the maximum sum-rate ($d_{RD}^\textrm{opt}$) does not change by varying the transmit power.
     As we can see from this figure, given $d_{SD}=18$\,km,
     a channel with $P_{t}=14.5$ dB, $\kappa^{u}=30$ dB and $\kappa^{l}=15$ dB
     offers maximum sum-rate, when $d_{RD}^\textrm{opt}=15.5$ km.
     Also, when $\kappa^{u}=30$ dB and $\kappa^{l}=20$ dB, the optimal relay's altitude is
     14 km. In fact, as the difference between $\kappa^{u}$ and $\kappa^{l}$
     become smaller, $d_{RD}^\textrm{opt}$ converges towards $\frac{d_{SD}}{2}$. \vspace{0.02cm}
\end{MyColorPar}
	\section{Conclusion}\vspace{-0.00cm}
	
	In this paper, we have proposed an effective interference alignment scheme for maximizing
    the sum-rate of HAPs-ground stations communications assisted by a tethered balloon
    relay. In particular, 
    we have considered the half-duplex relaying scheme using a tethered balloon relay to achieve maximum DoF in HAPs-ground stations communications,
    when the HAPs lack the knowledge of CSI.
Our results have shown that, using a tethered balloon relay for exploiting IA provides a significant sum-rate gain in the HAP-based wireless system that uses multiple interfering drones. \vspace{0.2cm}

	\def\baselinestretch{1.04}
	\bibliographystyle{IEEEtran}
\vspace{-0.3cm}
		\bibliography{references}
	\vspace{-0.3cm}
\end{document}